\documentclass[3p,12pt]{elsarticle}
\usepackage{longtable}
\usepackage{color}
\usepackage{colortbl}
\usepackage{hyperref}
\usepackage{amssymb}
\usepackage{hhline,color,colortbl,cancel}

\biboptions{sort}
\journal{Discrete Mathematics}

\newtheorem{theorem}{Theorem}
\newtheorem{lemma}{Lemma}
\newtheorem{proposition}{Proposition}
\newtheorem{corollary}{Corollary}
\newtheorem{remark}{Remark}

\newenvironment{proof}[1][\hspace{-1.0ex}]%
{\addvspace{1mm}{\bf Proof\hspace{1.0ex}{#1}.} }%
{\quad$\square$\par\addvspace{1mm}}

\begin{document}
\begin{frontmatter}
\title{The minimum volume of subspace trades}
\tnotetext[mytitlenote]{The results were presented in part at the International Conference
``Mal'tsev Meeting''
dedicated to 75th anniversary of Yuri L. Ershov,
May 3--7, 2015, Novosibirsk, Russia.}
\tnotetext[mytitlenote]{This is the author's version of the paper published
in Discrete Mathematics, Volume 340, Issue 12, 2017, Pages 2723-2731, DOI
\href{https://doi.org/10.1016/j.disc.2017.08.012}{10.1016/j.disc.2017.08.012} 
\copyright 2017 Elsevier B.V.  }

\author[mymainaddress]{Denis S. Krotov}
\ead{krotov@math.nsc.ru}

\address[mymainaddress]{Sobolev Institute of Mathematics, pr. Akademika Koptyuga 4, Novosibirsk 680090, Russia}

\begin{abstract}
 A subspace bitrade of type $T_q(t,k,v)$ is a pair $(T_0,T_1)$ of two disjoint nonempty collections of $k$-dimensional subspaces of a $v$-dimensional space $V$ over the finite field of order $q$ such that every $t$-dimensional subspace of $V$ is covered by the same number of subspaces from $T_0$ and $T_1$. In a previous paper, the minimum cardinality of a subspace $T_q(t,t+1,v)$ bitrade was established. We generalize that result by showing that for admissible $v$, $t$, and $k$, the minimum cardinality of a subspace $T_q(t,k,v)$ bitrade does not depend on $k$. An example of a minimum bitrade is represented using generator matrices in the reduced echelon form. For $t=1$, the uniqueness of a minimum bitrade is proved.
\end{abstract}

\begin{keyword}
Bitrades \sep Trades \sep Subspace designs
\MSC[2010] 05B30 \sep  05E99
\end{keyword}

\end{frontmatter}

\newcommand{\GrSS}[1]{\mathcal{F}^v_{#1}}
\newcommand{\GRss}[1]{{J}_q(v,#1)}
\newcommand{\HatH}[1]{\mathcal{H}(#1)}

\begin{abstract}
 A subspace bitrade of type $T_q(t,k,v)$ 
 is a pair $(T_0,T_1)$
 of two disjoint nonempty collections 
 of $k$-dimensional subspaces 
 of a $v$-dimensional space $V$ 
 over the finite field of order $q$
 such that every $t$-dimensional subspace of $V$
 is covered by the same number of subspaces
 from $T_0$ and $T_1$.
 In a previous paper, the minimum cardinality of a subspace $T_q(t,t+1,v)$ bitrade
 was established.
 We generalize that result by showing that for admissible $v$, $t$, and $k$,
 the minimum cardinality of a subspace $T_q(t,k,v)$ bitrade does not depend on $k$.
 An example of a minimum bitrade is represented using generator matrices in the reduced echelon form.
 For $t=1$, the uniqueness of a minimum bitrade is proved.
\end{abstract}

\section{Introduction}

Trades (see, e.g., \cite{HedKho:trades}) are traditionally used for constructing and studying combinatorial designs.
The relation of trades and designs can be briefly described as follows: if a design includes a trade subset, then this subset can be replaced by another (mate) trade to form a new design with the same parameters. 
This operation, also known as switching (see, e.g., \cite{Ost:2012:switching}), 
is widely studied not only for designs but for many different kinds of combinatorial objects (latin squares and hypercubes, codes, etc).
In view of the growing interest to the $q$-ary (subspace) generalizations of designs in the last few years, 
the study of corresponding analogs of trades becomes actual. 
Subspace trades are already used in the construction of subspace designs \cite{BKKL:large}.

In the current paper, we establish the minimum possible cardinality of a subspace ${T}_q(t,k,v)$ trade corresponding to the $q$-ary generalizations of ${S}(t,k,v)$ designs. 
The case ${T}_q(t,t+1,v)$ was solved in \cite{Cho:99}.
In \cite{Cho:99}, 
the equivalent language of null designs was used instead of trades,
and the result for the subspace trades appears as a partial case of more general theory of trades (null designs) in ranked posets.
In \cite{KMP:16:trades}, the same partial result 
(on the minimum subspace trades with $k=t+1$) 
was represented in a different general context, 
in terms of so-called clique trades in distance-regular graphs with regular systems of Delsarte cliques.
(Regretfully, the authors of \cite{KMP:16:trades} did not refer \cite{Cho:99} properly 
as they were not familiar with the theory of null designs.
As one of them, the author of the current paper is also responsible for this inconvenience.)
We will use the solution for $k=t+1$ to prove the general result for subspace trades.
As in \cite{KMP:16:trades}, the proof of the bound in the current paper exploits 
a weight distribution of the characteristic function of a bitrade 
(a bitrade, also known as a $2$-way trade, is a pair of mate trades).

There is one essential difference between the  ${T}_q(t,t+1,v)$ trades and the general case.
In the first case, the characteristic function of a bitrade is an eigenfunction of the Grassman graph, 
and the weight distribution with respect to a single vertex works well for establishing a lower bound on the number of nonzeros of the function.
In the general case, the eigenspectrum of the
bitrade characteristic function consists of several eigenvalues. 
To neutralize all components but one,
we modify the technique and count the weight distribution with respect to a special completely regular set. 
The invariance of the considered distributions allows to prove the tight lower bound on the cardinality of a bitrade, 
but at the moment, 
does not allow to generalize the distance-regular properties of minimum ${T}_q(t,t+1,v)$ 
trades that were established in \cite{KMP:16:trades}.

Another difference is that for $k=t+1$ a minimum ${T}_q(t,k,v)$ trade is a Steiner trade, 
that is, every $t$-subspace is covered by at most one $k$-subspace from the trade
(this follows from the explicit construction and the uniqueness of such a trade \cite{Cho:98}).
This is not necessarily the case in general; 
the known example of a minimum trade (see Section~\ref{ss:example}) is not Steiner if $k>t+1$.
The minimum size of a Steiner subspace ${T}_q(t,k,v)$ trade, $k>t+1$, remains unknown. 

Many results on subspace designs can be treated as analogs (sometimes, essentially more complicated) 
of similar facts in the theory of ordinary designs, and our result is not an exception.
The minimum volume $2^t$ of the ordinary trade was established in \cite{FranklPach:83} and~\cite{Hwang:86}.
Our method can be considered as the development of the method of \cite{FranklPach:83}; 
however, we use an alternative terminology of the theory of distance-regular graphs instead 
of the terminology of ranked posets, which was developed in \cite{Cho:99} in this context.
It would be interesting to find a subspace analog of the elegant inductive proof 
of the lower bound $2^t$ on the minimum trade volume found in~\cite{Hwang:86}.

%


In Section~\ref{s:pre} we introduce the notation and the main theorem.
Section~\ref{s:proof} contains a proof of the result.
In Section~\ref{s:concl} we discuss some open problems, including the uniqueness of minimum subspace trades 
(we prove it for $t=1$). 

\section{Notation and the main theorem}\label{s:pre}
\begin{itemize}
  \item $q$ is a prime power; $v$ is an integer, $v\ge 4$.
  \item ${\mathbb F}^v$ --- a $v$-dimensional space over the finite field ${\mathbb F}=\mathrm{GF}(q)$ of order $q$.
  \item $\GrSS{i}$ --- the set of all $i$-dimensional subspaces of ${\mathbb F}^v$, $i\in\{0,\ldots,v\}$.
  \item $\GRss{i}$ --- the \emph{Grassmann graph} on the vertex set $\GrSS{i}$; 
  two subspaces $X$, $Y\in \GrSS{i}$ are adjacent if $\dim(X\cap Y) = i-1$.
  \item ${d}(X,Y)$ --- the natural graph distance between two vertices $X$ and $Y$ of the graph. The distance from a vertex to a set of vertices is defined in a usual way:
  ${d}(X,\mathcal{Y}):=\min_{Y\in \mathcal{Y}} {d}(X,Y)$.
  \item $t$, $k$ --- integers satisfying $0\le t < k < v-t$.
  \item A pair $(\mathcal{T}_0,\mathcal{T}_1)$ of disjoint nonempty multisubsets of $\GrSS{k}$ 
is called a ${T}_q(t,k,v)$ \emph{subspace bitrade} if every subspace from $\GrSS{t}$
is covered by the same number of subspaces from $\mathcal{T}_0$ and $\mathcal{T}_1$
(the readers who are not interested in trades with repetitions can imply that $\mathcal{T}_0$ and $\mathcal{T}_1$ are ordinary sets, without multiplicities).
  \item We will refer to the value $|\mathcal T_0 \cup \mathcal T_1|$ as the \emph{cardinality} of a 
  bitrade $(\mathcal{T}_0,\mathcal{T}_1)$, while the value $|\mathcal T_0|$ is known as its \emph{volume} (from the definition, it follows that $|\mathcal T_0| = |\mathcal T_1|$).
\end{itemize}

Our goal is to prove the following theorem.
\begin{theorem}\label{th:asdiy}
Let the integers $t$, $k$, and $v$ satisfy $0\le t<k<v-t$. The 
minimum possible cardinality $|\mathcal T_0 \cup \mathcal T_1|$ of a subspace bitrade $(\mathcal T_0,\mathcal T_1)$ of type ${T}_q(t,k,v)$ equals\\[-2.0ex]
\begin{equation}
 \label{eq:main}
\prod_{i=0}^t(1+q^i)=\sum_{j=0}^{t+1} q^{\frac{j(j-1)}2}{\textstyle\big[{t+1 \atop j}\big]}_q
 \end{equation}
where 
$\big[{i \atop j}\big]_q :=\displaystyle \frac{[i]_q[i-1]_q\ldots [i-j+1]_q}{[1]_q[2]_q\ldots [j]_q}$,\quad $[r]_q:=1+q+\ldots+q^{r-1}$.
\end{theorem}

An example of a minimum bitrade is given in Section~\ref{ss:example}.
Formula (\ref{eq:main}) is a known identity \cite[Equation (1.87)]{Stanley};
in the proof, we will refer to its right part.
Theorem~\ref{th:asdiy} is the subspace analog of the similar result for the classical design trades 
(where the role of ${\mathbb F}^v$ is played by a set of $v$ elements,
the role of the $i$-dimensional subspaces is played by the $i$-subsets, and the formulas hold with $q=1$)
 proved in \cite{FranklPach:83,Hwang:86}.

The next group of definitions and notations concerns the space of real-valued functions on $\GrSS{k}$.
Such functions will be denoted by lowercase Greek letters; the only exception is $\theta$, which will always denote an eigenvalue.

\begin{itemize}

\item For a subset $\mathcal{S}$ 
of the vertex set $\mathcal{V}$ 
of a graph, we denote
$$ \mathcal{S}^{(i)} := \{ Y\in \mathcal{V} \mid {d}(Y,\mathcal{S})=i\}. $$
\item A real-valued function $\varphi$ 
on the set $\mathcal{V}$ 
of vertices of a simple graph 
is called an \emph{eigenfunction}
with \emph{eigenvalue} $\theta$
if for all $X$ from $\mathcal{V}$ it holds
$$ \sum_{Y\in \{X\}^{(1)}} \varphi(Y) = \theta \varphi(X).$$

\item Two real-valued functions 
$\varphi$, $\psi$ on $\mathcal{V}$ 
are 
\emph{orthogonal}, $\varphi\perp\psi$,
if  
$\displaystyle \sum_{X\in \mathcal{V}} \varphi(X)\psi(X) = 0.$

\item $ \bar k:= \min(k,v-k)$ --- the diameter of $\GRss{k}$

\item $\theta_0$, $\theta_1$, \ldots, $\theta_{\bar k}$ 
--- the eigenvalues of $\GRss{k}$;
we assume $\theta_0>\theta_1> \ldots >\theta_{\bar k}$.

\item $\Theta_0$, $\Theta_1$, \ldots, $\Theta_{\bar k}$ 
--- the eigenspaces corresponding to the eigenvalues 
$\theta_0$, $\theta_1$, \ldots, $\theta_{\bar k}$,
respectively.

So, $\Theta_0+\Theta_1+\ldots+\Theta_{\bar k}$ 
is the space of all real-valued functions over 
$\GrSS{k}$.

\item Given $X\in \GrSS{i}$, 
where $i\le k$, we denote 
$\HatH{X}:=\{Y \in \GrSS{k} \mid X\subseteq Y\}$. Note that this notation does not reflect the dependence on $k$, as $k$ is fixed.

\item  $\chi_\mathcal{C}$ 
denotes the multiplicity function 
of a multiset $\mathcal{C}$ of elements from $\GrSS{k}$ 
(which is the characteristic $\{0,1\}$-function if $\mathcal{C}$ is a set).

\item For a pair 
$(\mathcal C_0,\mathcal C_1)$ 
of multisubsets of $\GrSS{k}$, 
denote 
$\chi_{(\mathcal{C}_0,\mathcal{C}_1)}:=\chi_{\mathcal{C}_0}-\chi_{\mathcal{C}_1}$.

\item  $\Lambda_i:=\big\langle \big\{ \chi_{\HatH{X}} \mid X\in\GrSS{i} \big\}\big\rangle$, $i\le k$, where $\langle\ldots\rangle$ denotes the linear span.

\item The \emph{weight distribution} 
of a real-valued function $\varphi$ on $\GrSS{k}$ 
with respect to a subset $\mathcal C$ 
of $\GrSS{k}$ is the sequence 
$$ \overline W_\varphi=\left(W_\varphi^0,W_\varphi^1,\ldots,W_\varphi^{r}\right),\quad
\mbox{where}\quad r=\max_{\mathcal C^{(j)}\ne\emptyset} j  \quad \mbox{and}\quad
 W_\varphi^j = \sum_{Y\in\mathcal C^{(j)}} \varphi(Y).$$

\end{itemize}

\section{Proof}\label{s:proof}
In this section, we will prove Theorem~\ref{th:asdiy}. 
The final proof is contained in Subsection~\ref{ss:th},
while in the Subsections~\ref{ss:eig}--\ref{ss:wd} 
we consider auxiliary lemmas.
\subsection{Eigenspaces of a Grassmann graph}\label{ss:eig}

In the next lemma, we remind known facts about the structure of the eigenspaces of a Grassmann graph.
A detailed proof for the case $2k\le n$ can be found in \cite{MP:Norton}.
For completeness, we include arguments for the general case.

\begin{lemma}\label{l:chb}
(i) Assume $j \le \bar k$. Then
$ \Lambda_j = \Theta_0+\Theta_1+\ldots+\Theta_{j}$; 
if $j>0$, then $\Theta_j = \Lambda_j \cap \Lambda_{j-1}^\perp$. 

(ii) Assume $\bar k\le j \le k$. Then $\Lambda_j = \Lambda_k$.
\end{lemma}
\begin{proof}
  Claim (i) is proved in \cite{MP:Norton} for the case $\bar k = k$ (claim (ii) is trivial in this case).
  
  Some statements from \cite{MP:Norton} work also for the case $\bar k < k$, with the same arguments.
  In particular,
  for all $j$ from $1$ to $k$, we have 
  \begin{equation}\label{eq:sklisa}
     \Lambda_{j-1} \subseteq \Lambda_j
  \end{equation}
  (this fact is rather simple as for any $X$ from $\GrSS{j-1}$ it holds 
  $$\chi_{\HatH{X}}=c\hspace{-1em} \sum_{Y\in\GrSS{j}: X\subset Y}\hspace{-1em} \chi_{\HatH{Y}}$$ 
  for some constant $c$) and, 
  moreover, the space $\Lambda_j \cap \Lambda_{j-1}^\perp$ is a subspace of an eigenspace of the graph
  \cite{MP:Norton}.
  
  If we have (ii), 
  then $\Lambda_0$ and $\Lambda_j \cap \Lambda_{j-1}^\perp$, 
  $j=1,\ldots,\bar k$ span the full space of real-valued functions 
  on $\GrSS{k}$. 
  Since $\GRss{k}$ is a distance-regular graph of diameter $\bar k$ \cite[9.3]{Brouwer},
  it has exactly $\bar k+1$ distinct eigenvalues \cite[4.1.B]{Brouwer}; 
  so, all these spaces are eigenspaces. 
  It follows from the multiplicities of the eigenvalues
  ($\theta_i$ has the multiplicity $\big[{v \atop i}]_q-\big[{v \atop i-1}]_q$  \cite[Theorem~9.3.3]{Brouwer})
  that the corresponding eigenvalues are $\theta_0$, $\theta_1$, \ldots, $\theta_{\bar k}$, respectively.
  
  It remains to prove (ii), which, taking into account (\ref{eq:sklisa}), is equivalent to $\Lambda_{\bar k} = \Lambda_k$.
  It is sufficient to show that for every $Y$ from $\GrSS{k}$, the function 
  $\chi_{\{Y\}}=\chi_{\HatH{Y}}$ is a linear combination of $\chi_{\HatH{X}}$, $X\in\GrSS{\bar k}$.
  For such $Y$, consider the $\bar k+1$ functions
  on $\GrSS{k}$
  $$\psi_i =\hspace{-2em} \sum_{X\in\GrSS{\bar k}:\,\mathrm{dim}(X\cap Y)=\bar k - i}\hspace{-2em}\chi_{\HatH{X}},\qquad 
  i=0,\ldots,\bar k,$$
  and the $\bar k+1$ sets
  $\mathcal{Y}^{(j)}$, $j=0,\ldots,\bar k$,
  where $\mathcal{Y}=\{Y\}$.
  It is straightforward that for every $i,j\in\{0,\ldots,\bar k\}$,
  \begin{itemize}
  \item  $\psi_i$ is constant on  $\mathcal{Y}^{(j)}$;
  \item if $i>j$, then $\psi_i$ is constantly zero on  $\mathcal Y^{(j)}$;
  \item if $i=j$, then $\psi_i$ is non-zero on  $\mathcal{Y}^{(j)}$.
  \end{itemize}
  So, $\psi_i = a_{ii}\chi_{\mathcal{Y}^{(i)}}
  +
  \ldots
  +
  a_{i\bar k}\chi_{\mathcal{Y}^{(\bar k)}}$, $a_{ii}\ne 0$, $i=0,\ldots,\bar k$.
  
  It is easy to see that each $j\in\{0,\ldots,\bar k\}$,
   the function $\chi_{\mathcal{Y}^{(j)}}$
  is a linear combination of $\psi_i$, $i=j,\ldots,\bar k$.
  In particular, this holds for $\chi_{\mathcal{Y}^{(0)}}=\chi_{\{Y\}}$.
\end{proof}

\subsection{A definition of bitrades in terms of eigenspaces}\label{ss:alt-def}

The next lemma gives an alternative definition of a subspace bitrade $(\mathcal T_0,\mathcal T_1)$ in terms
of the eigenspectrum of   $\chi_{(\mathcal T_0,\mathcal T_1)}$.
 
\begin{lemma}\label{l:juce}
  A pair $(\mathcal T_0,\mathcal T_1)$ of two disjoint multisubsets of $\GrSS{k}$ is a ${T}_q(t,k,v)$ bitrade if and only if 
  $$\chi_{(\mathcal T_0,\mathcal T_1)} \perp \Theta_j\qquad \mbox{for all $j=0,1,\ldots,t$}.$$
\end{lemma}
\begin{proof}
  Indeed, for every $X\in\GrSS{t}$, the relation $\chi_{(\mathcal T_0,\mathcal T_1)} \perp \chi_{\HatH{X}}$ 
  is equivalent to the fact that $\mathcal T_0$ and $\mathcal T_1$ have the same number of elements covering $X$.
  So, $(\mathcal T_0,\mathcal T_1)$ is a bitrade if and only if 
  \begin{equation}\label{eq:perp}
    \chi_{(\mathcal T_0,\mathcal T_1)} \perp \Lambda_t.
  \end{equation}
  Applying Lemma~\ref{l:chb} finishes the proof.
\end{proof}

We now see why the condition $k< v-t$ is necessary: if $k\ge v-t$ (equivalently, $\bar k \le t$), 
then $\Lambda_t$ is the space of all functions and (\ref{eq:perp}) yields
$\chi_{(\mathcal T_0,\mathcal T_1)}\equiv 0$.
\begin{corollary}
Subspace ${T}_q(t,k,v)$ bitrades 
do not exist if $k\ge v-t$.
\end{corollary}
 \subsection{An example of a minimum subspace bitrade}\label{ss:example}
 
In the following lemma, 
we consider properties of a concrete subspace bitrade, 
which was discovered in \cite[Ch.\,18]{James:84} for $k=t+1$, in an alternative terminology,
and the generalization for $k>t+1$ was mentioned in \cite{Cho:98} 
(there were not any details given, 
but the idea how to lift an example from $k=t+1$ to $k>t+1$ is evident)
As will be proved later, this bitrade is minimum by cardinality.
The proof of the lemma is based on the partial case $k=t+1$ of the lemma statement, 
which was proved earlier \cite{James:84,KMP:16:trades}.
Because of the importance of this case, 
we give an alternative proof of the required properties of the bitrade ($k=t+1$) in  Appendix,
utilizing its representation via reduced row echelon matrices.
\begin{lemma}\label{l:duh}
Assume $0<t<k<v-t$. 
Let $\mathcal T\subset \GrSS{k}$, 
consist of all $k$-dimensional subspaces of ${\mathbb F}^v$ whose all vectors satisfy
\begin{equation}\label{eq:ex-min-tr}
 x_1x_{2t+2}+ x_2x_{2t+1}+\ldots + x_{t+1}x_{t+2}=0,\qquad x_{k+t+2}=\ldots=x_v=0,
\end{equation}
written as $v$-tuples $(x_1,\ldots,x_v)$ in some fixed basis $(e_1,\ldots,e_v)$.
Then $\mathcal T$ is partitioned into two independent sets $\mathcal T_0$, $\mathcal T_1$ and satisfies the following
assertions:

(A) $(\mathcal T_0,\mathcal T_1)$ is a subspace ${T}_q(t,k,v)$ bitrade.

(B) $ \displaystyle|\mathcal T|=|\mathcal T_0+\mathcal T_1|=\sum_{j=0}^{t+1} q^{\frac{j(j-1)}2}\big[{\textstyle {t+1 \atop j}}\big]_q$.

(C) Denote by $Z'$ the $(t+1)$-dimensional subspace of ${\mathbb F}^v$ 
whose all vectors satisfy $x_{1}=\ldots=x_{t+1}=x_{2t+3}=\ldots=x_{v}=0$.
The weight distribution of $\chi_{(\mathcal T_0,\mathcal T_1)}$ with respect to 
${\HatH{Z'}}$ is 
\begin{equation}\label{eq:wd-ex}
 \left((-1)^j q^{\frac{j(j-1)}2}\bigg[{t+1 \atop j}\bigg]_q\right)_{j=0}^{t+1}
\mbox{\ \ or\ \ } 
\left(-(-1)^{j} q^{\frac{j(j-1)}2}\bigg[{t+1 \atop j}\bigg]_q\right)_{j=0}^{t+1}.
\end{equation}
\end{lemma}

\begin{proof}
 (A) Every subspace $Y$ from $\mathcal T$ meets $\langle e_{2t+3},\ldots,e_{k+t+1}\rangle\subset Y \subset \langle e_{1},\ldots,e_{k+t+1}\rangle$. 
 We represent every element $Y$ of $\mathcal T$ in the form
 $$ Y = Y' + \langle e_{2t+3},\ldots,e_{k+t+1}\rangle,\qquad\mbox{where } Y' = Y \cap \langle e_1,\ldots,e_{2t+2}\rangle.$$
 By $\mathcal T'$, we denote the set $\{Y' \mid Y\in \mathcal T\}$.
 
  It is known \cite{KMP:16:trades} (see also the Appendix, Lemma~\ref{l:k-1}) 
 that $\mathcal T'$ is splittable into a ${T}_q(t,t+1,v)$ bitrade $(\mathcal T'_0,\mathcal T'_1)$, $\mathcal T'_0 \cup\mathcal T'_1=\mathcal T'$. 
 It follows from Lemma~\ref{l:juce} (see also  \cite[Lemma~4.3]{BKKL:large}) 
 that $(\mathcal T'_0,\mathcal T'_1)$ is 
 a ${T}_q(s,t+1,v)$ bitrade for every $s\le t$.
 Now consider any subspace $X$ from $\GrSS{t}$.
 If $X\not\subset \langle e_1,\ldots,e_{k+t+1}\rangle$, 
 then, obviously, it is not included in any subspace from $\mathcal T$.
 If $X\subset \langle e_1,\ldots,e_{k+t+1}\rangle$, then we denote by $X'$ the projection of $X$ onto 
 $\langle e_1,\ldots,e_{2t+2}\rangle$. Denote the dimension of $X'$ by $s$.
 It is obvious that for every $Y$ from $\mathcal T$, $X \subset Y$ if and only if $X' \subset Y'$.
 Since $(\mathcal T'_0,\mathcal T'_1)$ is a ${T}_q(s,t+1,v)$ bitrade, it follows that $X$ is included in the same number of subspaces 
 from $\mathcal T_0$ and from $\mathcal T_1$.
 
 (B) The cardinality of $\mathcal T$ equals the cardinality of  $\mathcal T'$, which is known \cite{KMP:16:trades} (see also the Appendix)
 (actually, the elements of $\mathcal T'$ are the vertices of the subgraph of the Grassmann graph known as the dual polar graph of type $[D_d(q)]$, $d=t+1$, see e.g. \cite[9.4]{Brouwer}).
 
 (C) Let us first prove that 
 \emph{the weight distribution of 
 $\chi_{(\mathcal T_0,\mathcal T_1)}$ 
 with respect to 
${\HatH{Z'}}$ is equal to 
the weight distribution of 
 $\chi_{(\mathcal T'_0,\mathcal T'_1)}$ 
 with respect to 
$\{Z'\}$}.
Indeed, given 
$Y = Y' + \langle e_{2t+3},\ldots,e_{k+t+1}\rangle$ from 
$\mathcal{T}$, 
the element of $\HatH{Z'}$
nearest to $Y$ is obviously
$Z = Z' + \langle e_{2t+3},\ldots,e_{k+t+1}\rangle$.
So, $d(Y,\HatH{Z'}) = d(Y,Z)=d(Y',Z')$,
which means that the contribution of $Y$
to the weight distribution of 
 $\chi_{(\mathcal T_0,\mathcal T_1)}$ 
 with respect to 
${\HatH{Z'}}$ is the same as the contribution of $Y'$
to the weight distribution of
 $\chi_{(\mathcal T'_0,\mathcal T'_1)}$ with respect to 
$\{Z'\}$.
 
 The last weight distribution is known \cite{KMP:16:trades} (see also the Appendix, Lemma~\ref{l:5}) and equals one of (\ref{eq:wd-ex}),
 the first formula corresponding to the case 
 $Z'\in \mathcal{T}'_0$, the last, to $Z'\in \mathcal{T}'_1$.
 \end{proof}

 \subsection{Completely regular sets and weight distributions}\label{ss:wd}
 
A set $\mathcal{S}$ of  vertices of a connected graph is said to be 
\emph{completely regular} if for any $i$ and $j$, 
all vertices from $\mathcal{S}^{(i)}$ 
have the same number $s_{i,j}$ 
of neighbors in $\mathcal{S}^{(j)}$.
The numbers 
$s_{i,i-1}$,
$s_{i,i}$,
$s_{i,i+1}$,
$i=0,1,\ldots$,
are referred to as the
\emph{intersection numbers}
of the completely regular set $\mathcal{S}$.

\begin{lemma}%
[see, e.g., \cite{Kro:struct}]
\label{l:wd-calc}
Let $\varphi$ be an eigenfunction, with the eigenvalue $\theta$,
of a connected graph,
and let $\mathcal{S}$ 
be a completely regular set of vertices 
of the graph
with the intersection numbers
$s_{i,i-1}$, $s_{i,i}$, $s_{i,i+1}$, $i=0,1,\ldots$.
Then 
the weight distribution
$(W_\varphi^0,\ldots,W_\varphi^r)$ of $\varphi$
with respect to
$\mathcal{S}$
equals
\begin{equation}
\label{eq:W}
W_\varphi^0\cdot\left(w^0,\ldots,w^r\right),
\quad \mbox{where }
w^{i+1} = 
\frac{\theta w^{i}
- s_{i,i}w^{i}
- s_{{i{-}1},i}w^{i-1}}{s_{{i{+}1},i}}
, \  i=0,1,\ldots,r-1, \quad 
\end{equation}
$w^{-1}=0,\ w^0=1.$
\end{lemma}
\begin{proof}[(a sketch)]
Consider the identity 
$$\sum_{X\in \mathcal{S}^{(i)}}
\theta\varphi(X)=\sum_{X\in \mathcal{S}^{(i)}}
\sum_{Z\in \{X\}^{(1)}}\varphi(Z)$$
and note that the value 
$\varphi(Z)$, $Z\in \mathcal{S}^{(j)}$, 
is included $s_{j,i}$
times in the right part. This gives a recursive formula for $W_\varphi^{i}$,
which is equivalent to (\ref{eq:W}).
\end{proof}

The next lemma follows from direct calculations, as well as from the symmetry.

\begin{lemma}\label{l:H-cr}
 For any $X\in \GrSS{j}$, $j\in \{0,1,\ldots,k\}$
 the set $\HatH{X}$ is completely regular. 
 The corresponding intersection numbers depend on $j$ but do not depend on the choice of $X$ in $\GrSS{j}$.
\end{lemma}

\begin{lemma}\label{l:ckbckbh}
Let the integers $t$, $k$, and $v$ satisfy $0\le t<k<v-t$.
Let $(\mathcal T_0,\mathcal T_1)$ 
be a subspace bitrade of type ${T}_q(t,k,v)$,
and let $Z' \in \GrSS{t+1}$.
Then the weight 
distribution of 
$\chi_{(\mathcal T_0,\mathcal T_1)}$ with respect to $\HatH{Z'}$ is proportional to
{\rm (\ref{eq:wd-ex})}.
\end{lemma}
\begin{proof}
By Lemmas~\ref{l:chb} and~\ref{l:juce},
$\chi_{(\mathcal T_0,\mathcal T_1)}$ is orthogonal to 
$\Theta_0+\Theta_1+\ldots+\Theta_t$.
It follows that 
\begin{equation}
\label{eq:decomp}
\chi_{(\mathcal T_0,\mathcal T_1)}=\chi^{(t+1)}+\ldots+\chi^{(\bar k)}
\end{equation}
where $\chi^{(i)} \in \Theta_i$ is an eigenfunction corresponding to the eigenvalue $\theta_i$.

By Lemma~\ref{l:chb}, for every $i=t+2,\ldots,\bar k$, the eigenfunction $\chi^{(i)}$ is orthogonal to  $\chi_{\HatH{Z'}}$.
Hence, by Lemma~\ref{l:wd-calc}, the weight distribution of $\chi^{(i)}$ with respect to $\HatH{Z'}$ is the all-zero tuple for all $i=t+2,\ldots,\bar k$.
It follows that the weight distribution of $\chi_{(\mathcal T_0,\mathcal T_1)}$ (with respect to $\HatH{Z'}$) equals 
the weight distribution of $\chi^{(t+1)}$ and equals $W_{\chi^{(t+1)}}^0(w^0,\ldots,w^r)$, see (\ref{eq:W}).
Note that $w^0,\ldots,w^r$ do not depend on the choice of $Z'$, as the intersection numbers of $\HatH{Z'}$ are the same.
By the example of the bitrade considered in Section~\ref{ss:example}, we see that $(w^0,\ldots,w^r)$ has the form  (\ref{eq:wd-ex}).
\end{proof}

\subsection{A proof of the main theorem}\label{ss:th}
\begin{proof}[of Theorem~\ref{th:asdiy}]
The upper bound is given by the example above.
Consider a ${T}_q(t,k,v)$ subspace bitrade $(\mathcal{T}_0,\mathcal{T}_1)$.
Let us prove that $|\mathcal{T}_0|+|\mathcal{T}_1|$ is not less than (\ref{eq:main}).
We proceed by induction on $\bar k - t-1$. 

(i) {\it The induction base}. 
In the case $\bar k- t-1=0$, i.e., $\bar k= t+1$, the decomposition (\ref{eq:decomp}) consists of one eigenfunction $\chi^{(t+1)}$,
which is not constantly zero. 
The last means that by Lemma~\ref{l:chb}
there is at least one $Z' \in \GrSS{t+1}$ such that $\chi^{(t+1)}$ is not orthogonal to $\chi_{\HatH{Z'}}$.
In particular, the weight distribution of $\chi^{(t+1)}$ 
(and hence, of $\chi_{(\mathcal T_0,\mathcal T_1)}$)
with respect to $\HatH{Z'}$ is nonzero and proportional to (\ref{eq:wd-ex}).
Since the first element of (\ref{eq:wd-ex}) is $1$ and the weight distribution consists of integers, 
the coefficient cannot be less that $1$, in absolute value. 
It follows that  $\mathcal T_0$ and $\mathcal T_1$
has at least $q^{\frac{j(j-1)}2}\big[{t+1 \atop j }\big]_q$ elements (taking into account the multiplicities) 
at distance $j$ from $\HatH{Z'}$.
So, the total number of elements cannot be less than (\ref{eq:main}).

(ii) {\it The induction step}. Assume that $\bar k > t+1$. 
If $\chi^{(t+1)}\equiv 0$ in (\ref{eq:decomp}), then by Lemma~\ref{l:juce} $(\mathcal{T}_0,\mathcal{T}_1)$ 
is a ${T}_q(t+1,k,v)$ subspace bitrade. Using the induction hypothesis, we conclude that 
$|\mathcal{T}_0|+|\mathcal{T}_1|$ is greater than (\ref{eq:main}) 
(it is easy to see that (\ref{eq:main}) is monotonic in $t$).

If $\chi^{(t+1)}$
 is not constantly zero, then the arguments in (i) also work, taking into account that 
 by Lemma~\ref{l:ckbckbh} the other summands of the decomposition (\ref{eq:decomp})
 add nothing to the weight distribution.
\end{proof}

\section{Concluding remarks}\label{s:concl}

We established the minimum cardinality of a ${T}_q(t,k,v)$ bitrade.

The question if a minimum ${T}_q(t,k,v)$ bitrade is unique remains open.  
The uniqueness in the case $k=t+1$ was established by Cho~\cite{Cho:98}.  
Another related result was obtained by Pankov \cite{Pankov:Dd-to-Gr}:
it was shown that embedding of the dual polar graph $[D_d(q)]$ in the Grassmann graph is unique, up to isomorphism
(as follows from the definition \cite[9.4]{Brouwer}, the minimum bitrade described in Section~\ref{ss:example} induces 
a dual polar subgraph $[D_d(q)]$ in the Grassmann graph).
In the ordinary case, bitrades of minimum volume are not unique if $k>t+1$. 
For example, $(\{\{1,2,3,4\},\{5,6,7,8\}\},\{\{1,2,3,8\},\{5,6,7,4\}\})$ and $(\{\{1,2,3,4\},\{5,6,7,8\}\},\{\{1,2,7,8\},\{5,6,3,4\}\})$
are non-isomorphic $T(1,4,8)$ trades. 
However, the situation with the subspace $T_q(1,k,v)$ trades is different, 
and it is natural to conjecture that it is so for any $t$:
\begin{proposition}\label{p:unique}
 For each $k\ge 2$, $v\ge k+2$, and prime power $q$, there is only one $T_q(1,k,v)$ trade, up to isomorphism.
\end{proposition}
\begin{proof}
 It is sufficient to show that all the elements of a trade intersect in one $(k-2)$-dimensional subspace $Z$.
 In this case, factorization by $Z$ leads to a  $T_q(1,2,v-k+2)$ trade, a case solved in~\cite{Cho:98}.
 
 Let $(\mathcal T_0,\mathcal T_1)= (\{X_0, \ldots, X_q\},\{Y_0, \ldots, Y_q\})$ be a $T_q(1,k,v)$ trade.
 The $k$-dimensional space $X_0$ has $(q^k-1)/(q-1)$ one-dimensional subspaces.
 Since $((q^k-1)/(q-1))/(q+1)> (q^{k-2}-1)/(q-1)$, at least one of $q+1$ elements of $\mathcal T_1$, 
 say $Y_0$, covers $(q^{k-1}-1)/(q-1)$
 of those $1$-dimensional subspaces. 
 Next, each of $Y_i$, $i=1,\ldots,q$ covers at most $(q^{k-1}-1)/(q-1)$ one-dimensional subspaces of $X_0$,
 and at most $(q^{k-1}-1)/(q-1) - (q^{k-2}-1)/(q-1) = q^{k-2}$ of them are ``new'', i.e., not covered by $Y_0$, \ldots, $Y_{i-1}$.
 Since $Y_0$ does not cover $q^{k-1}=q\cdot q^{k-2}$ one-dimensional subspaces of $X_0$, this number is precisely 
 $q^{k-2}$. If follows that
  
  (i) $\dim(X_0\cup Y_i)=k-1$, $i=0,\ldots,q$; 

  (ii) $\dim(X_0\cup Y_i\cup Y_j)=k-2$, $i\ne j$; 
  
  (iii) $X_0\cup Y_0\cup Y_1 = X_0\cup Y_0\cup Y_i$, $i=2,\ldots,q$.
  
  We see that all $Y$s intersect in a same $(k-2)$-dimensional subspace $Z$, which is also a subspace of $X_0$. 
  Similarly, $Z$ is a subspace of $X_1$, \ldots, $X_q$.
\end{proof}

Another question is if there are subspace designs 
including a minimum trade as a subset.
What is known is that among the ${S}_q(2,3,13)$ systems found in \cite{BEOVW:q-Steiner}
(which are currently the only known nontrivial ${S}_q(t,k,v)$ systems with $t\ge 2$)
several tested systems do not include minimum trades
\cite{Alfred:2015}.

Finally, we mention another variant of the minimum-volume problem of a subset trade.
In \cite{Cho:98}, it is conjectured that the minimum number of \emph{different} elements in a subspace 
$T_q(t,k,v)$ trade $(\mathcal T_0, \mathcal T_1)$  
is (\ref{eq:main}) for $k>t+1$ (for $k=t+1$, the problem is solved). 
The result of the current paper confirms the conjecture for the case when 
$\mathcal T_0$, $\mathcal T_1$ are ordinary sets (as the number 
of different elements coincides with the cardinality in this case), 
but in general case, with repetitions allowed, similar arguments do not work,
and the problem remains unsolved, with the exception of the simple case $t=1$
(see the proof of Proposition~\ref{p:unique}).

\section*{Appendix. The structure of a minimum trade in the case $k=t+1$}

As was mentioned above, in the case $k=t+1$, 
the properties (A), (B) and (C) of the trade considered in Lemma~\ref{l:duh}
are already known.
However, the proof of these properties was based on a more general theory, 
which is not reasonable to be described in the current paper.
For completeness, we give here a direct proof of these properties, 
based on the generator matrices in the reduced row echelon form.
An $l\times n$ matrix of rank $l$ is called a \emph{reduced row echelon} matrix
if 
the leading coefficient (the first nonzero element from the left) 
of a row is always $1$, 
it is to the right of the leading coefficients of all rows above,
and it is an only nonzero element of the corresponding column.
For a given subspace, a reduced row echelon generator matrix always exists, unique, 
and can be found from any generator matrix by Gauss--Jordan elimination 
(the basis of the space is considered to be fixed).

Let us consider a subspace 
$U$ of a $2k$-dimensional space $V$
such that all vectors  $\bar x = (x_1,\ldots ,x_{2k})$ of $U$ meet
\begin{equation}\label{eq:0} 
Q(\bar x)=0, \qquad\mbox{where } Q(\bar x) := x_1x_{2k}+x_2x_{2k-1}+\ldots +x_k x_{k+1}.
\end{equation}

For $\bar x =(x_1,\ldots ,x_{2k})$ and $\bar y =(y_1,\ldots ,y_{2k})$, denote 
\begin{equation}\label{eq:1}
 \lfloor \bar x, \bar y \rceil := x_1 y_{2k}+x_2y_{2k-1}+\ldots+x_{2k} y_1.
\end{equation}

\begin{lemma}\label{l:1}
 All $\bar x$ and $\bar y$ from $U$ satisfy 
$
 \lfloor \bar x, \bar y \rceil = 0.
$
\end{lemma}
\begin{proof}
 It is easy to see that 
$
 Q(\bar x+\bar y) = Q(\bar x) + \lfloor \bar x, \bar y \rceil + Q(\bar y).
$
Since $\bar x+\bar y$, $\bar x$, $\bar y\in U$, we have 
$Q(\bar x+\bar y) = Q(\bar x) = Q(\bar y) = 0$. 
Hence, $\lfloor \bar x, \bar y \rceil$ is $0$ too.
\end{proof}

Let $\bar u^1$, \ldots , $\bar u^l$ be the rows of the reduced row echelon generator matrix $\bar U$ of $U$.
Let $r_i$ be the position of the leading one (the first nonzero element) in $\bar u^i$.
\begin{lemma}\label{l:2}
For all $i$ and $j$ from $1$ to $\dim(U)$, we have $r_i \ne 2k+1-r_j$.
In other words, each of the pairs $\{ 1,2k\}$, $\{2,2k-1\}$, \ldots, $\{k,k+1\}$ 
contains at most one leading one of $\bar U$.
\end{lemma}
\begin{proof}
Seeking a contradiction, assume $r_i = 2k+1-r_j$.
Then, there is only one $l$ (namely, $l=r_i$) such that $\bar u^i$ has nonzero in the $l$th position
and $\bar u^j$  has nonzero in the $(2k+1-l)$th position.
From (\ref{eq:1}), 
where $\bar x =\bar u^i$ and  $\bar y =\bar u^j$,
we see that $\lfloor \bar u^i, \bar u^j \rceil = 1$, 
which contradicts Lemma~\ref{l:1}.
\end{proof}
\begin{corollary}\label{c:1}
 The dimension of $U$ does not exceed $k$.
 If $\dim(U)=k$, then each of the pairs 
$\{ 1,2k\}$, $\{2,2k-1\}$, \ldots, $\{k,k+1\}$
intersects with $\{r_1,\ldots,r_k\}$ 
in exactly one element.
\end{corollary}
Denote by $u^{i}_{j}$ the value of the $(2k+1-r_j)$th position of $\bar u^i$.
Note that if $\dim(U)=k$, then $\bar U$ does not have nonzero elements other than $u^i_j$s and leading ones, see e.g. Fig.~\ref{f:k}.
\begin{figure}
$$
\definecolor{Green}{rgb}{0,0.5,0}
\definecolor{gray}{rgb}{0.5,0.5,0.5}
\begin{array}{|>{\columncolor[gray]{0.95}}cc>{\columncolor[gray]{0.95}}cc>{\columncolor[gray]{0.95}}c>{\columncolor[gray]{0.95}}cc|>{\columncolor[gray]{0.95}}ccc>{\columncolor[gray]{0.95}}cc>{\columncolor[gray]{0.95}}cc|}
\hline
1 & \color{blue}u^1_7 & 0 & \color{blue}u^1_6 & 0 & 0 & \color{blue}u^1_5 & 0 & \color{blue}u^1_4 & \color{blue}u^1_3 & 0 & \color{blue}u^1_2 & 0 & \color{gray}\cancel{u^1_1} 
\\ \rowcolor[gray]{0.9} &\cdot&\cdot&\cdot&\cdot&\cdot&\cdot&\cdot&\cdot&\cdot&\cdot&\cdot&\cdot&\cdot \\
  &       & 1 & \color{blue}u^2_6 & 0 & 0 & \color{blue}u^2_5 & 0 & \color{blue}u^2_4 & \color{blue}u^2_3 & 0 & \color{gray}\cancel{u^2_2}& 0 & \color{Green}u^2_1 
\\ \rowcolor[gray]{0.9} &&&\cdot&\cdot&\cdot&\cdot&\cdot&\cdot&\cdot&\cdot&\cdot&\cdot&\cdot  \\
  &       &   &       & 1 & 0 & \color{blue}u^3_5 & 0 & \color{blue}u^3_4 & \color{gray}\cancel{u^3_3} & 0 & \color{Green}u^3_2 & 0 & \color{Green}u^3_1 \\
  &       &   &       &   & 1 & \color{blue}u^4_5 & 0 & \color{gray}\cancel{u^4_4} & \color{Green}u^4_3 & 0 & \color{Green}u^4_2 & 0 & \color{Green}u^4_1 
\\ \rowcolor[gray]{0.9} &&&&&&\cdot&\cdot&\cdot&\cdot&\cdot&\cdot&\cdot&\cdot \\ \hline
  &       &   &       &   &   &       & 1 & \color{Green}u^5_4 & \color{Green}u^5_3 & 0 & \color{Green}u^5_2 & 0 & \color{Green}u^5_1 
\\ \rowcolor[gray]{0.9} &&&&&&&&\cdot&\cdot&\cdot&\cdot&\cdot&\cdot 
\\ \rowcolor[gray]{0.9} &&&&&&&&&\cdot&\cdot&\cdot&\cdot&\cdot \\
  &       &   &       &   &   &       &   &       &       & 1 & \color{Green}u^6_2 & 0 & \color{Green}u^6_1 
\\ \rowcolor[gray]{0.9} &&&&&&&&&&&\cdot&\cdot&\cdot \\
  &       &   &       &   &   &       &   &       &       &   &       & 1 & \color{Green}u^7_1 
\\ \rowcolor[gray]{0.9} &&&&&&&&&&&&&\cdot \\ \hline
\end{array}
\qquad { u^i_j = -u^j_i, \atop u^i_i = 0.}
$$
\caption{An example of $\bar U$ in the case $\dim(U)=k$. The empty grayed lines are inserted to emphasize some symmetry with respect to the secondary diagonal.}\label{f:k}
\end{figure}
\begin{lemma}\label{l:3}
Assume that $\dim(U)=k$.
Then 
for all  $i$ and $j$ from $1$ to $k$, we have $u^{i}_{j}=-u^{j}_{i}$; moreover, $u^{i}_i=0$.
\end{lemma}
\begin{proof}
Consider (\ref{eq:1}), 
where $\bar x =\bar u^i$ and  $\bar y =\bar u^j$.
If $l\ne r_i,2k+1-r_j$, 
then $x_l y_{2k+1-l} = 0$ 
(indeed, either $l$ 
is the position of the leading one of some $\bar u^s\ne \bar u^i$,
or $2k+1-l$ is the position of the leading one of some $\bar u^s\ne \bar u^{j}$).

So, (\ref{eq:1}) becomes 
$ x_{r_i} y_{2k+1-r_i} + x_{2k+1-r_j} y_{r_j}$
if $i\ne j$, and 
$ x_{r_i} y_{2k+1-r_i}$
if $i=j$.
Since $ x_{r_i} = y_{r_j}=1$, $y_{2k+1-r_i} = u^j_i$, and $x_{2k+1-r_j} = u^i_j$,
we find from Lemma~\ref{l:1} that $ u^j_i+u^i_j=0$ and $ u^i_i=0$.
\end{proof}

\begin{lemma}\label{l:4}
Let $\bar U$ be a  reduced row echelon $k\times 2k$ matrix such that 

(A)
the leading ones $r_1$, \ldots, $r_k$ 
occur once in each of the pairs 
$\{ 1,2k\}$, $\{2,2k-1\}$, \ldots, $\{k,k+1\}$;

(B) for all $i$, $j$ from $1$ to $k$, it holds $u^{i}_{j}=-u^{j}_{i}$ and  $u^{i}_{j}=0$,
in the notation of Lemma~\ref{l:3}.

Then,  all vectors $\bar x = (x_1,\ldots ,x_{2k})$ of the space $U$ spanned by the rows of $\bar U$ meet (\ref{eq:0}).
\end{lemma}
\begin{proof}
 From the proof of Lemma~\ref{l:3} 
 we see that $\lfloor \bar u^i, \bar u^j \rceil=0$
 for all $i$, $j$ from $1$ to $k$.
 Similarly, $Q(\bar u^i)=0$.
 Now,
 
 $\displaystyle Q(\alpha_1 \bar u^1+\ldots+\alpha_k \bar u^k) = 
 \sum_{i=1}^k \alpha_i^2 Q(\bar u^i) + \sum_{i=1}^{k-1}\sum_{j=i+1}^{k}\alpha_i\alpha_j \lfloor \bar u^i, \bar u^j \rceil=0
 .$
\end{proof}
For a subspace $U$,
denote by $s(U)$ the number of rows with leading ones from $1$ to $k$ in the  reduced row echelon matrix generating $U$.
Note that there is a unique subspace $U_0$ of dimension $k$ with $s(U_0)=0$, 
and for any other subspace $U$ of dimension $k$, it holds $s(U)=d(U,U_0)$.
\begin{lemma}\label{l:5}
For $s\in \{0,\ldots ,k\}$, the number of the matrices $U$ meeting the hypothesis of Lemma~\ref{l:4} 
and satisfying $s(U)=s$ is
 $q^{s(s-1)/2}\left[{k\atop s}\right]_q$.
\end{lemma}
\begin{proof}
The upper-left $s\times k$ submatrix $U^{\mathrm{UL}}$ of $U$ is an arbitrary $s\times k$  reduced row echelon matrix.
The number of such matrices is $\left[{k\atop s}\right]_q$ (the number of $s$-dimensional subspaces of a $k$-dimensional space).

The lower-right $(k-s)\times k$  submatrix of $U$ if uniquely determined by $U^{\mathrm{UL}}$, by Lemmas~\ref{l:2} and~\ref{l:3}.

The lower-left $(k-s)\times k$ submatrix consists of zeros.

The upper-right $s\times k$  submatrix of $U$ consists of $s(k-1)$ zeros and the elements $u^i_j$, where $i,j\le s$.
From these elements, $s(s-1)/2$ can be chosen arbitrarily (for example, with $i<j$); the rest is uniquely determined by 
Lemma~\ref{l:3}.
\end{proof}
\begin{lemma}\label{l:k-1}
 Assume that $\dim(U)=k-1$. 
 Then, $U$ is included in exactly two subspaces $W$ of dimension $k$ whose all vectors meet~(\ref{eq:0}), 
 one with $s(W)=s(U)$ and one with $s(W)=s(U)+1$.
\end{lemma}
\begin{proof}
 Consider the  reduced row echelon matrix $\bar U$ generating $U$. 
 As before, $r_1$, \ldots, $r_{k-1}$ denote the positions of the leading ones of the rows of $\bar U$.
 Consider a subspace $W$ satisfying the conditions of the lemma.
 Then $W$ is spanned by the rows of $\bar U$ and an additional vector $\bar w$.
 We can assume that $\bar w$ has zeros in the positions $r_1$, \ldots, $r_{k-1}$
 and has a leading one in some position $r'$. 
 By Lemma~\ref{l:2}, 
 only one pair of $\{ 1,2k\}$, $\{2,2k-1\}$, \ldots, $\{k,k+1\}$ does not contain one of
 $r_1$, \ldots, $r_{k-1}$, and this pair contains $r'$.
 So, there are only two possibilities to choose $r'$, one (when $r'>k$) for $s(W)=s(U)$ and
 one (when $r' \le k$) for $s(W)=s(U)+1$. It remains to show that for each of these two cases,
 $\bar w$ is uniquely determined. 
 Indeed, in the positions $r_i$, $i=1,\ldots ,k-1$, the vector $\bar w$ has zeros, the other positions (see examples at Fig.~\ref{f:k-1}):
 
\begin{figure}
$$
\definecolor{LightCyan}{rgb}{0.88,1,1}
\begin{array}{r|cccc>{\columncolor{LightCyan}}ccc|cc>{\columncolor{LightCyan}}ccccc|l}
\hhline{~|-------|-------|}
\bar u_1: &1 & u^1_6 & 0 & u^1_5 & y^1 & 0 & u^1_4 & 0 & u^1_3 & x^1 & 0 & u^1_2 & 0 & {u^1_1} &\ \  \mbox{\underline{$s(W)=s(U)+1$}} \\
\bar u_2: &  &       & 1 & u^2_5 & y^2 & 0 & u^2_4 & 0 & u^2_3 & x^2 & 0 & {u^2_2}& 0 & u^2_1 \\
\bar u_3: &  &       &   &       &   & 1 & u^3_4 & 0 & {u^3_3} & x^3 & 0 & u^3_2 & 0 & u^3_1 \\ \hhline{~|-------|-------|}
\bar u_4: &  &       &   &       &   &   &       & 1 & u^4_3 & x^4 & 0 & u^4_2 & 0 & u^4_1 \\
&  &       &   &       &   &   &       &   &       &       & 1 & u^5_2 & 0 & u^5_1 \\
&  &       &   &       &   &   &       &   &       &       &   &       & 1 & u^6_1
& \ \  Q(\bar w) {=} 0 \Rightarrow w_0 {=} 0
\\ \hhline{~|-------|-------|}
\bar w:&  &       &   &       & 1 & 0 & w_4 & 0 & w_3 & \cancel{w_0} & 0 & w_2 & 0 & w_1
& \ \  \lfloor \bar w, \bar u_i\rceil {=} 0 \Rightarrow w_i{=}{-x^i}

  \end{array}
$$
$$
\definecolor{LightCyan}{rgb}{0.88,1,1}
\begin{array}{r|cccc>{\columncolor{LightCyan}}ccc|cc>{\columncolor{LightCyan}}ccccc|l}
\hhline{~|-------|-------|}
\bar u_1: &1 & u^1_6 & 0 & u^1_5 & y^1 & 0 & u^1_4 & 0 & u^1_3 & x^1 & 0 & u^1_2 & 0 & {u^1_1} &\ \  \mbox{\underline{$s(W)=s(U)$}} \\
\bar u_2: &  &       & 1 & u^2_5 & y^2 & 0 & u^2_4 & 0 & u^2_3 & x^2 & 0 & {u^2_2}& 0 & u^2_1 \\
&  &       &   &       &   & 1 & u^3_4 & 0 & {u^3_3} & x^3 & 0 & u^3_2 & 0 & u^3_1 \\ \hhline{~|-------|-------|}
&  &       &   &       &   &   &       & 1 & u^4_3 & x^4 & 0 & u^4_2 & 0 & u^4_1 \\
&  &       &   &       &   &   &       &   &       &       & 1 & u^5_2 & 0 & u^5_1 \\
&  &       &   &       &   &   &       &   &       &       &   &       & 1 & u^6_1
& 
\\ \hhline{~|-------|-------|}
\bar w:&  &       &   &       &  &  &  &  &  & 1 & 0 & w_2 & 0 & w_1
& \ \  \lfloor \bar w, \bar u_i\rceil {=} 0 \Rightarrow w_i{=}{-y^i}

  \end{array}
$$
\caption{An example of the case $\dim(U)=k-1$: completing to dimension $k$.}\label{f:k-1}
\end{figure}
 
 a) $r_i$, $i=1,\ldots ,k-1$, in which $\bar w$ has $0$;
 
 b) $r'$, in which $\bar w$ has $1$;
 
 c) $2k+1-r'$, in which $\bar w$ has $0$, as $Q(\bar w)=0$;
 
 d) $2k+1-r_i$, $i=1,\ldots ,k-1$, in which the value of $\bar w$ is uniquely determined by 
 $\lfloor \bar w, \bar u_i \rceil = 0$ 
 (if $2k+1-r_i < r'$, then this value is $0$ as $\bar w$ and $\bar u_i$ has no common nonzero positions).
 
 After we uniquely (for each $r'$) determine $\bar w$ from the rules above, we have $Q(\bar w)=0$ and 
 $\lfloor \bar w, \bar u_i \rceil = 0$, $i=1,\ldots ,k-1$.
 By the arguments similar to the proof of Lemma~\ref{l:4}, all vectors of $W$ meet~(\ref{eq:0}). 
\end{proof}

We conclude the appendix with the note that the statement of Lemma~\ref{l:duh} for $k=t+1$ 
is straightforward from 
Lemmas~\ref{l:5} and~\ref{l:k-1}. 
To see this, assume without loss of generality that $v=2t+2$ 
(the zero coordinates with larger indices do not play any role), 
set $\mathcal T$ to consist of all $k$-dimensional subspaces $U$ of $V$ 
whose all vectors satisfy (\ref{eq:0}), and split it to 
$\mathcal T_0$, $\mathcal T_1$ in accordance to the parity of $s(U)$.

\begin{remark}
It is not difficult to ``lift'' the representation via reduced row echelon matrices from the case 
$k=t+1$, $v=2t+2$ to an arbitrary case considered in Lemma~\ref{l:duh}. 
To do this, one should add $k-t-1$ new rows and $v-2t-2$ new columns, all filled with zeros 
except $k-t-1$ leading ones of the new rows, in positions from 
$2t+3$ to $k+t+1$.
\end{remark}

\section*{Acknowledgement} 

The author thanks Professor Dennis W. Stanton for noticing the papers \cite{Cho:98,Cho:99}.
The work was funded by the Russian Science Foundation (Grant 14-11-00555).

\smallskip\smallskip
\providecommand\href[2]{#2} \providecommand\url[1]{\href{#1}{#1}}
  \def\DOI#1{{\small {DOI}:
  \href{http://dx.doi.org/#1}{#1}}}\def\DOIURL#1#2{{\small{DOI}:
  \href{http://dx.doi.org/#2}{#1}}}

\end{document}